\documentclass[12pt,a4paper]{article}
\usepackage[english]{babel}

\usepackage[a4paper,top=2cm,bottom=2cm,left=2.5cm,right=2.5cm,marginparwidth=1.75cm]{geometry}

\usepackage{amsmath}
\usepackage{graphicx}
\usepackage{hyperref}
\usepackage[title]{appendix}
\usepackage{mathrsfs}
\usepackage{amsfonts}
\usepackage{caption}
\usepackage{algorithm}
\usepackage{algorithmicx}
\usepackage{algpseudocode}
\usepackage{listings}
\usepackage{enumitem}
\usepackage{chngcntr}
\usepackage{booktabs}
\usepackage{lipsum}
\usepackage{subcaption}
\usepackage{authblk}
\usepackage{diagbox}
\usepackage{comment}
\usepackage[T2A]{fontenc}
\usepackage{bbold}
\usepackage{dsfont}
\usepackage{amsthm}

\newcommand\cA{{\mathcal A}}
\newcommand\cF{{\mathcal F}}
\newcommand\cI{{\mathcal I}}
\newcommand\cM{{\mathcal M}}
\newcommand\cD{{\mathcal D}}
\newcommand\cP{{\mathcal P}}
\newcommand\cZ{{\mathcal Z}}
\newcommand\cW{{\mathcal W}}

\newcommand\E{{\mathbb E}}

\newcommand\1{{\mathbb 1}}
\newcommand\e{{\varepsilon}}

\newcommand\dP{{\mathds{P}}}

\newcommand{\sgn}{\mathop{\text{sign}}}

\newtheorem{theorem}{Theorem}[section]
\newtheorem{proposition}{Proposition}[section]
\theoremstyle{definition}
\newtheorem*{definition}{Definition}

\DeclareMathOperator*{\argmax}{arg\,max}

\title{Implicit Numerical Scheme for the Hamilton--Jacobi--Bellman Quasi-Variational Inequality in the Optimal Market-Making Problem with Alpha Signal}

\author{Alexey Meteykin}
\affil{\small Lomonosov Moscow State University, Vega Institute Foundation, Moscow, Russia}

\date{}  

\begin{document}
\maketitle

\begin{center}
    \begin{minipage}{0.9\textwidth}
        \noindent\textbf{Abstract:} We address the problem of combined stochastic and impulse control for a market maker operating in a limit order book. The problem is formulated as a Hamilton–Jacobi–Bellman quasi-variational inequality (HJBQVI). We propose an implicit time-discretization scheme coupled with a policy iteration algorithm. This approach removes time-step restrictions typical of explicit methods and ensures unconditional stability. Convergence to the unique viscosity solution is established by verifying monotonicity, stability, and consistency conditions and applying the comparison principle.
        \medskip
        
        \noindent\textbf{Keywords:} Hamilton--Jacobi--Bellman equation, combined stochastic and impulse control, implicit numerical scheme, policy iteration, viscosity solution.
    \end{minipage}
\end{center}


\newpage
\section{Introduction}

Optimal control problems play an important role in modern financial mathematics. One such problem is the optimal market making problem in a limit order book, which naturally arises in high-frequency trading on electronic exchanges.

A market maker, as a market participant, can submit two types of orders:
\begin{itemize}
    \item Limit order --- an instruction to buy or sell an asset at a given or more favorable price. Such an order is posted to the limit order book and is executed when it matches an opposite market order.
    \item Market order --- an instruction to immediately buy or sell an asset at the best available price in the limit order book. Market orders enable rapid position adjustments but are executed at less favorable prices compared to limit orders.
\end{itemize}

The market maker provides liquidity to the market by placing limit buy and sell orders, and, when necessary, uses market orders to manage risk. The profit of the market maker arises from the spread between the buy and sell prices. The main source of risk is the potential loss in inventory value due to adverse price movements, known as inventory risk. To mitigate this risk, the market maker aims to maintain a near-zero inventory position.

Classical market making models, such as \cite{avellanedastoikov} and \cite{gueantlehalletapia}, assume that price dynamics contain no predictable components and that the market maker has no predictive information about future price movements. In \cite{carteawang}, this classical setup is extended by introducing an alpha signal, a stochastic process representing the predictable component of price dynamics. This generalization allows for the possibility of extracting additional profit from predictable price trends.

In \cite{carteawang}, the authors formulate the market making problem with an alpha signal, derived from the information contained in the flow of market orders, as a combined stochastic and impulse control problem. They show that the optimal market making strategy satisfies a Hamilton--Jacobi--Bellman Quasi-Variational Inequality (HJBQVI), which is solved numerically using an explicit time discretization scheme. However, the explicit scheme imposes natural restrictions on the time step due to stability constraints, which in turn increase the computational cost of the numerical solution.

In this paper, we propose an alternative numerical approach to solving the above HJBQVI. The main idea of the proposed method is to use an implicit time discretization scheme combined with a policy iteration algorithm. The implicit scheme is unconditionally stable and therefore eliminates restrictions on the time step size, while the policy iteration algorithm efficiently solves the discrete equation at each time step.

We establish convergence of the proposed implicit numerical scheme to the unique viscosity solution of the HJBQVI by applying the general framework developed in \cite{barlessouganidis}. Specifically, we verify the conditions of monotonicity, stability, and consistency, as well as the validity of the comparison principle for the original equation. The convergence of the policy iteration algorithm is justified using the results of \cite{azimzadehforsyth}, which provide sufficient conditions for convergence in terms of diagonal dominance of the associated matrices and the connectivity properties of the graph induced by the impulse control component.

\section{Model}

Let $(\Omega, {\cF}, \bf F, \dP)$ be a stochastic basis, where the filtration ${\bf F} = \{\cF_t\}_{t\in[0,T]}$ is generated by the processes $W$, $\bar{M}^{\operatorname{a}}$, and $\bar{M}^{\operatorname{b}}$ introduced below. The market maker operates over the interval  $[0, T]$ with a fixed trading horizon $T>0$.

The midprice in the order book is defined as the average of the best bid and best ask prices. Let the midprice $S = (S_t)_{t\geq 0}$ of the asset evolve as 
\begin{equation*}
    dS_t = \sigma (dJ^\uparrow_t - dJ^\downarrow_t),
\end{equation*}
where $\sigma > 0$ denotes the minimum price increment (tick size) in the order book. Here $J^\uparrow$ and $J^\downarrow$ are doubly stochastic Poisson processes with intensities 
\begin{equation*}
    \mu^\uparrow_t = \alpha_t^+ + \theta \quad \text{and} \quad \mu^\downarrow_t = \alpha_t^- + \theta,
\end{equation*}
where $\theta > 0$ is a constant and $\alpha = (\alpha_t)_{t\geq 0}$ is a process, referred to as the alpha signal and  capturing information about the directional component of price dynamics, will be specified below. 

The market maker may submit both limit and market orders, all having unit size.

The control for limit sell orders is an ${\bf F}$-predictable process $l^{\operatorname{a}} =  (l_t^{\operatorname{a}})_{t \in [0, T]}$ with values in $\{0, 1\}$. A limit sell order is posted at time $t$ if $l_t^{\operatorname{a}} = 1$ and is not posted otherwise. Limit sell orders are placed at the best ask price $S_t + \Delta$, where $\Delta \geq 0$ represents half of the bid–ask spread. Analogously, the process $l^{\operatorname{b}} =  (l_t^{\operatorname{b}})_{t \in [0, T]}$ controls limit buy orders, which are placed at the best bid price $S_t - \Delta$.

Let $N^\text{a} = (N^\text{a}_t)_{t\geq0}$ and $N^\text{b} = (N^\text{b}_t)_{t\geq0}$ denote the counting processes for executed limit sell and buy orders of the market maker, respectively.

A limit sell (buy) order posted by the market maker is assumed to be executed with probability one whenever a market buy (sell) order arrives in the market.

The control governing market orders is specified by a double sequence
\begin{equation*}
    \zeta=(\tau_1, \tau_2, \dotsc; z_1, z_2, \dotsc),
\end{equation*}
where $0 \leq \tau_1 \leq \tau_2 \leq \cdots$ are ${\bf F}$-stopping times, and $z_1, z_2, \dotsc \in \{1, -1\}$ are impulses corresponding to these times. At time $\tau_i$, the market maker submits a market buy order if $z_i = 1$, and a market sell order if $z_i = -1$. Market buy orders are executed at $S_t + \Upsilon$, and market sell orders at $S_t - \Upsilon$, where $\Upsilon = \Delta + \e$ denotes the total cost of taking liquidity, with $\e > 0$ being the market order fee. Let $M^{\operatorname{a}} = (M^{\operatorname{a}}_t)_{t\geq0}$ and $M^{\operatorname{b}} = (M^{\operatorname{b}}_t)_{t\geq0}$ denote the counting processes for the market maker’s buy and sell market orders, respectively.

Other participants also submit market orders. Let $\bar{M}^{\operatorname{a}} = (\bar{M}^{\operatorname{a}}_t)_{t\geq 0}$ denote the Poisson process counting external market buy orders with intensity $\lambda^{\operatorname{a}}$, and $\bar{M}^{\operatorname{b}} = (\bar{M}^{\operatorname{b}}_t)_{t\geq 0}$ the Poisson process counting external market sell orders with intensity $\lambda^{\operatorname{b}}$.

The alpha signal $\alpha = (\alpha_t)_{t\geq 0}$ evolves as an Ornstein--Uhlenbeck process between jump times of the processes $\bar{M}_t^{\operatorname{a}}, M_t^{\operatorname{a}}, \bar{M}_t^{\operatorname{b}}, M_t^{\operatorname{b}}$:
\begin{equation*}
        d\alpha_t = -k\alpha_t dt + \rho dW_t + \gamma^{\operatorname{a}} (d\bar{M}_t^{\operatorname{a}} + dM_t^{\operatorname{a}}) - \gamma^{\operatorname{b}} (d\bar{M}_t^{\operatorname{b}} + dM_t^{\operatorname{b}}), \quad \alpha_0 = 0,
\end{equation*}
where $W = (W_t)_{t \geq 0}$ is a Brownian motion, and $k, \rho, \gamma^{\operatorname{a}}, \gamma^{\operatorname{b}} > 0$ are constants. Each market buy order arrival increases $\alpha_t$ by $\gamma^{\operatorname{a}}$, while each market sell order decreases it by $\gamma^{\operatorname{b}}$. Thus, the alpha signal reflects the imbalance between buy and sell market order flows.

We denote the market maker’s control by $\nu = (l^{\operatorname{a}}, l^{\operatorname{b}}, \zeta)$.

The controlled inventory process $Q^\nu = (Q^\nu_t)_{t\geq 0}$ is given by the relation 
\begin{equation*}
        Q^\nu_t = N^\text{b} - N^\text{a} + M^\text{a}_t - M^\text{b}_t; 
\end{equation*}
the controlled cash process $X^\nu = (X^\nu_t)_{t\geq0}$ evolves as
\begin{equation*}
        dX^\nu_t = (S_t + \Delta) dN^\text{a}_t - (S_t - \Delta) dN^\text{b}_t - (S_t + \Upsilon)dM^\text{a}_t + (S_t - \Upsilon)dM^\text{b}_t.
\end{equation*}

The set of admissible controls $\cA$ includes all $\nu = (l^{\operatorname{a}}, l^{\operatorname{b}}, \zeta)$ such that the inventory remains bounded, that is $Q^\nu_t \in [-\overline{Q}, \overline{Q}]$ for all $t \in [0, T]$ for some integer $\overline{Q} > 0$, and the impulse control $\zeta$ does not trigger simultaneous buy and sell market orders.

We define the state process $Y^\nu = \big\{(X_t^\nu, S_t, \alpha_t, Q_t^\nu)\big\}_{t \in [0, T]}$, and use the shorthand $y = (x, s, \alpha, q)$. The performance is measured by the  functional 
\begin{gather*}\label{cost_functional}
    J^\nu(t, y) = \E^{t, y}\bigg[\int_t^T f(r, Y_r^\nu)dr + g(T, Y_T^\nu)\bigg],
    \\
    f(t, y) = - \phi q^2, \quad g(t, y) = x + q(s - \Upsilon \sgn q) - \psi q^2,
\end{gather*}
where $\E^{t, y}[\cdot]$ denotes conditional expectation given $X^\nu_{t_-}=x$, $S_{t_-}=s$, $\alpha^\nu_{t_-}=\alpha$, $Q^\nu_{t_-}=q$, and $\psi, \phi > 0$ are constants.

The running cost $f(t, y) = -\phi q^2$ penalizes nonzero inventory levels, reflecting the exposure to adverse price movements. The terminal reward $g(t, y)$ represents the liquidation value at the terminal date, consisting of the current cash balance $x$ and the proceeds $q(s - \Upsilon \sgn q) - \psi q^2$ obtained upon liquidating the remaining position with a market order. The quadratic term $-\psi q^2$ accounts for additional costs due to insufficient liquidity at the best available quote in the order book.

The optimization problem is formulated as a combined stochastic and impulse maximization problem with the value function 
\begin{equation}\label{optimization_problem}
        u(t, y) = \sup_{\nu \in \cA} J^\nu(t, y).
\end{equation}

The dimension of the problem can be reduced to three variables by applying the substitution
\begin{equation}\label{substitution}
    u(t,x,s,\alpha,q)=x+qs+v(t,\alpha,q),
\end{equation}
where the value function $u(t,x,s,\alpha,q)$ decomposes into three components: the accumulated cash $x$, the current mark-to-market value $qs$ of the inventory, and the residual function $v(t,\alpha,q)$ representing the expected additional profit generated on $[t, T]$ under the optimal strategy.

As shown in \cite{carteawang}, the value function (\ref{optimization_problem}) is the unique viscosity solution to the corresponding Hamilton--Jacobi--Bellman Quasi-Variational Inequality (HJBQVI), which, under substitution (\ref{substitution}), takes the form
\begin{equation}\label{hjb}
    \begin{cases}
        \max\Bigg(
        \sup\limits_{l_t^{\operatorname{a}}, l_t^{\operatorname{b}} \in \{0, 1\}} \Big(\frac{\partial v}{\partial t} + L^{l^{\operatorname{a}}, l^{\operatorname{b}}} v + \Tilde{f}^{l^{\operatorname{a}}, l^{\operatorname{b}}} \Big),
        \sup\limits_{z\in\{1, -1\}} \Big( \cM^z v - v \Big)
        \Bigg) = 0, & t \in [0, T),
        \\
        v(T, \alpha, q\big) = \Tilde{g}(T, \alpha, q\big),
    \end{cases}
\end{equation}
where
\begin{align*}
L^{l^{\operatorname{a}}, l^{\operatorname{b}}} v (t, \alpha, q\big) &= - k\alpha\frac{\partial v}{\partial \alpha}(t, \alpha, q\big) +\frac12 \rho^2\frac{\partial^2 v}{\partial\alpha^2}(t, \alpha, q\big)\\
&+ \lambda^{\operatorname{a}} \bigg(\1_{\{l^{\operatorname{a}}_t=0\}} v(t, \alpha+\gamma^{\operatorname{a}},q\big) + \1_{\{l^{\operatorname{a}}_t=1\}} v(t, \alpha+\gamma^{\operatorname{a}}, q-1\big) - v(t, \alpha, q\big)
\bigg)\\
&+ \lambda^{\operatorname{b}} \bigg(\1_{\{l^{\operatorname{b}}_t=0\}} v(t, \alpha-\gamma^{\operatorname{b}},q\big) + \1_{\{l^{\operatorname{b}}_t=1\}} v(t, \alpha-\gamma^{\operatorname{b}}, q+1\big) - v(t, \alpha, q\big) \bigg),\\
\cM^z v (t, \alpha, q\big) &=
\begin{cases}
    v(t, \alpha, q+1\big) -\Upsilon, & z=1,\\
    v(t, \alpha, q-1\big) -\Upsilon, & z=-1,
\end{cases}\\
\Tilde{f}^{l^{\operatorname{a}}, l^{\operatorname{b}}}(t, \alpha, q\big) &= \alpha\sigma q - \phi q^2 + \1_{\{l^{\operatorname{a}}_t=1\}} \lambda^{\operatorname{a}} \Delta + \1_{\{l^{\operatorname{b}}_t=1\}} \lambda^{\operatorname{b}} \Delta,\\
\Tilde{g}(t, \alpha, q\big) &= -\Upsilon q \sgn q - \psi q^2.
\end{align*}

\section{Numerical scheme}

Let $\{t_n\}_{n=0}^{N}$ be a uniform time grid with step size $\delta t > 0$ on the interval $[0, T]$, and let $\{\alpha_i\}_{i=-N_\alpha}^{N_\alpha}$ be a uniform grid with step size $\delta \alpha > 0$ on $[-A, A]$ for some $A > 0$. The inventory variable $q$ takes values in the  grid with unit spacing $\{q_j\}_{j=-N_q}^{N_q} = [-\overline{Q}, \overline{Q}] \cap \mathbb{Z}$.

Since the shifted values $\alpha + \gamma^{\operatorname{a}}$ and $\alpha - \gamma^{\operatorname{b}}$ may not coincide with grid points in $\alpha$, the value $v(t, \alpha, q)$ for $\alpha \in [-A, A] \setminus \{\alpha_i\}_{i=-N_\alpha}^{N_\alpha}$ is linearly interpolated between the nearest grid points. For $\alpha > A$, the value $v(t, \alpha, q)$ is linearly extrapolated using $v(t, A - \delta\alpha, q)$ and $v(t, A, q)$, and analogously for $\alpha < -A$. In equation (\ref{hjb}), the process $\alpha$ is replaced with its truncated version
\begin{equation*}
    \overline{\alpha}_t = \min\{A, \max\{-A, \alpha_t\}\}.
\end{equation*}

The partial derivatives are approximated by finite differences as follows:
\begin{align*}
    \frac{\partial v}{\partial t}(t, \alpha, q) &\sim \frac{v(t+\delta t, \alpha, q) - v(t, \alpha, q)}{\delta t},
    \\
    -k\alpha\frac{\partial v}{\partial \alpha}(t, \alpha, q) &\sim
    \begin{cases}
        -k\alpha\frac{v(t, \alpha + \delta \alpha, q) - v(t, \alpha, q)}{\delta \alpha}, &\alpha \leq 0,\\
        -k\alpha\frac{v(t, \alpha, q) - v(t, \alpha - \delta \alpha, q)}{\delta \alpha}, &\alpha \geq 0,
    \end{cases}
    \\
    \frac{\partial^2 v}{\partial\alpha^2}(t, \alpha, q) &\sim \frac{v(t, \alpha + \delta \alpha, q) - 2v(t, \alpha, q) + v(t, \alpha - \delta \alpha, q)}{\delta \alpha^2}.
\end{align*}

The boundary condition on the second derivative with respect to $\alpha$ is given by
\begin{equation*}
    \frac{\partial^2 v}{\partial\alpha^2}(t, -A, q) = 0 \quad \text{and} \quad \frac{\partial^2 v}{\partial\alpha^2}(t, A, q) = 0.
\end{equation*}

For convenience, denote $l^{\operatorname{a}}_n = l^{\operatorname{a}}_{t_n}$, $l^{\operatorname{b}}_n = l^{\operatorname{b}}_{t_n}$, $v^n=v(t_n, \alpha, q)$, and $\Tilde{f}^{l^{\operatorname{a}}, l^{\operatorname{b}}, n}=\Tilde{f}^{l^{\operatorname{a}}, l^{\operatorname{b}}}(t_n, \alpha, q)$.

The discrete form of the HJBQVI (\ref{hjb}) reads: 
\begin{equation}\label{scheme}
    \begin{cases}
        \max\Bigg(
        \sup\limits_{l_n^{\operatorname{a}}, l_n^{\operatorname{b}} \in \{0, 1\}} \Big(\frac{v^{n+1} - v^n}{\delta t} + L_{\delta}^{l^{\operatorname{a}}, l^{\operatorname{b}}} v^n + \Tilde{f}^{l^{\operatorname{a}}, l^{\operatorname{b}}, n} \Big),
        \sup\limits_{z\in\{1, -1\}} \Big( B_{\delta}^z v^n - \Upsilon \Big)
        \Bigg) = 0, & n < N,
        \\
        v(T, \alpha, q\big) = \Tilde{g}(T, \alpha, q\big),
    \end{cases}
\end{equation}
where
\begin{align*}
    L_{\delta}^{l^{\operatorname{a}}, l^{\operatorname{b}}} v(t, \alpha, q) &= k \overline{\alpha}_- \frac{v(t, \alpha + \delta \alpha, q) - v(t, \alpha, q)}{\delta \alpha} - k \overline{\alpha}_+ \frac{v(t, \alpha, q) - v(t, \alpha - \delta \alpha, q)}{\delta \alpha}
    \\
    &+ \frac{\rho^2}2 \frac{v(t, \alpha + \delta \alpha, q) - 2 v(t, \alpha, q) + v(t, \alpha - \delta \alpha, q)}{\delta \alpha^2}
    \\
    &+ \lambda^{\operatorname{a}} \bigg(\1_{\{l^{\operatorname{a}}_n=0\}}\cI^+ v(t, \alpha, q) + \1_{\{l^{\operatorname{a}}_n=1\}} \cI^+ v(t, \alpha, q - 1) - v(t, \alpha, q) \bigg)
    \\
    &+ \lambda^{\operatorname{b}} \bigg(\1_{\{l^{\operatorname{b}}_n=0\}}\cI^- v(t, \alpha, q) + \1_{\{l^{\operatorname{b}}_n=1\}} \cI^- v(t, \alpha, q + 1) - v(t, \alpha, q) \bigg),
    \\
    B_{\delta}^z v(t, \alpha, q) &= 
    \begin{cases}
        v(t, \alpha, q+1) -v(t, \alpha, q), & z=1,\\
        v(t, \alpha, q-1) -v(t, \alpha, q), & z=-1.
    \end{cases}
\end{align*}

The operators $\cI^+$ and $\cI^-$ perform linear interpolation with respect to $\alpha$:
\begin{align*}
    \begin{split}
        \cI^+ v(t, \alpha, q) & = v \bigg(t, \alpha + \bigg \lfloor \frac{\gamma^{\operatorname{a}}}{\delta \alpha}\bigg\rfloor \delta \alpha, q \bigg) \\
         & + \bigg( \frac{\gamma^{\operatorname{a}}}{\delta \alpha} - \bigg \lfloor \frac{\gamma^{\operatorname{a}}}{\delta \alpha}\bigg\rfloor \bigg) \Bigg( v \bigg( t, \alpha + \bigg \lceil \frac{\gamma^{\operatorname{a}}}{\delta \alpha}\bigg\rceil \delta \alpha , q \bigg) - v \bigg( t, \alpha + \bigg \lfloor \frac{\gamma^{\operatorname{a}}}{\delta \alpha}\bigg\rfloor \delta \alpha , q \bigg)\Bigg),
    \end{split}
    \\
    \begin{split}
        \cI^- v(t, \alpha, q) & = v \bigg(t, \alpha - \bigg \lfloor \frac{\gamma^{\operatorname{b}}}{\delta \alpha}\bigg\rfloor \delta \alpha, q \bigg) \\
        & + \bigg( \frac{\gamma^{\operatorname{b}}}{\delta \alpha} - \bigg \lfloor \frac{\gamma^{\operatorname{b}}}{\delta \alpha}\bigg\rfloor \bigg) \Bigg( v \bigg( t, \alpha - \bigg \lceil \frac{\gamma^{\operatorname{b}}}{\delta \alpha}\bigg\rceil \delta \alpha , q \bigg) - v \bigg( t, \alpha - \bigg \lfloor \frac{\gamma^{\operatorname{b}}}{\delta \alpha}\bigg\rfloor \delta \alpha , q \bigg)\Bigg).
    \end{split}
\end{align*}

\section{Policy iteration}

To solve the discrete equation (\ref{scheme}), we use the policy iteration algorithm. This algorithm addresses problems of the form
\begin{equation}\label{policy_iteration_problem}
    \sup_{P\in\cP}\bigg\{-A(P)v + b(P)\bigg\} = 0,
\end{equation}
where $A(P)$ is an $M \times M$ matrix, $b(P)$ and $v$ are vectors of length $M$, and $\cP$ denotes the set of admissible policies.
\begin{algorithm}
    \caption{Policy Iteration}\label{policy_iteration}
    \begin{algorithmic}[1]
        \State $r>0$ — tolerance level
        \State $v^0$ — initial guess
        \For{$k=0, 1, 2, \dots$}
            \State $P^k = \argmax\limits_{P\in\cP}\Big\{-A(P)v^k + b(P)\Big\}$
            \State Solve the linear system $A(P^k)v^{k+1}=b(P^k)$
            \If{$\max\limits_i\bigg|\dfrac{v^{k+1}_i - v^k_i}{v^{k+1}_i}\bigg|<r$}
                \State break
            \EndIf
        \EndFor
    \end{algorithmic}
\end{algorithm}

The terminal condition $v(T, \alpha, q) = \Tilde{g}(T, \alpha, q)$ of equation (\ref{scheme}) is specified at $t = T$. Proceeding backward in time, for each $n \in \{N-1, N-2, \dots, 1\}$, the solution $v^n$ is obtained by solving problem (\ref{policy_iteration_problem}) with the corresponding matrices $A(P)$ and vectors $b(P)$.

Fix $n \in \{N-1, N-2, \dots, 1\}$ and set $M = (2N_\alpha + 1)(2N_q + 1)$.

The set of admissible policies $\cP$ in the present problem is given by
\begin{equation*}
    \cP = \cW \times \cZ \times \cD,
\end{equation*}
where
\begin{equation*}
    \cW \subset \prod_{i=1}^M \Big( \{0, 1\}\times\{0, 1\} \Big),
    \quad
    \cZ \subset \prod_{i=1}^M\{1, -1\},
    \quad
    \cD = \prod_{i=1}^M\{0, 1\}.
\end{equation*}

Therefore, a policy $P = (w, z, d) \in \cP$ consists of three components. Namely,  the vector $w=(w_1, \dots, w_M)\in\cW$ corresponds to the stochastic control $(l^{\operatorname{a}}_n, l^{\operatorname{b}}_n)$ at each grid point, the vector $z=(z_1, \dots, z_M)\in\cZ$ represents the impulses, and the components of the vector $d=(d_1, \dots, d_M)\in\cD$ are indicators of impulse application. Let $D$ denote the diagonal matrix with $d=(d_1, \dots, d_M)$ on the diagonal.

To express equation (\ref{scheme}) in the form (\ref{policy_iteration_problem}), we write $A(P)$ and $b(P)$ as
\begin{align*}
    A(P) &=
        \big(I - D\big)\Big(I - L(w) \Big) + D\Big(I - B(z) \Big),
    \\
    b(P) &=
        \big(I - D\big)c(w) + Dk(z),
\end{align*}
where
\begin{align*}
    L(w) &= L_{\delta}^{l^{\operatorname{a}}, l^{\operatorname{b}}} \delta t,
    &
    c(w) &= v^{n+1} + \Tilde{f}^{l^{\operatorname{a}}, l^{\operatorname{b}}, n} \delta t,
    \\
    B(z) &= I + B_{\delta}^z,
    &
    k(z) &= -\Upsilon.
\end{align*}

To prove convergence of the policy iteration algorithm to the unique solution of equation (\ref{scheme}), we use the theorem from \cite{azimzadehforsyth}. For this purpose, several definitions are introduced below.

Let $A=(a_{ij})\in\mathbb{R}^{M\times M}$ be a real matrix.
\begin{definition}
The graph of a matrix $A$ is a graph with vertices $\{1, \dots, M\}$, where vertices $i$ and $j$ are connected by an edge if $a_{ij}\neq 0$.
\end{definition}
\begin{definition}
A matrix $A$ is a $Z$-matrix if $a_{ij}\leq 0$ for all $i\neq j$.
\end{definition}
\begin{definition}
A matrix $A$ is strictly (weakly) diagonally dominant if $|a_{ii}| > \sum_{j\neq i}|a_{ij}|$ ($|a_{ii}|\geq \sum_{j\neq i}|a_{ij}|$) for all $i$.
\end{definition}

\begin{theorem}[Convergence of policy iteration]\label{policy_iteration_convergence}
Assume that the following conditions hold:
\begin{enumerate}
    \item $P \mapsto A(P)^{-1}$ is bounded.
    \item $A$ and $b$ are bounded, and for every $x \in \mathbb{R}^M$ there exists a policy $P_x\in\cP$ such that $-A(P_x)x + b(P_x) = \sup_{P\in\cP}\{-A(P)x + b(P)\}$.
    \item For each $P=(w, z, d) \in \cP$ and vertex $i$ with $d_i=1$, there exists a path in the graph of the matrix $B(z)$ from $i$ to a vertex $j$ with $d_j=0$.
    \item For each $P=(w, z, d) \in \cP$, the matrices $I-L(w)$ and $I-B(z)$ are $Z$-matrices with nonnegative diagonal elements. The matrix $I-L(w)$ is strictly diagonally dominant, and the matrix $I-B(z)$ is weakly diagonally dominant.
\end{enumerate}
Then the sequence $(v^k)_{k=0}^{\infty}$ produced by the policy iteration algorithm (Algorithm \ref{policy_iteration}) is nondecreasing and converges to the unique solution $v$ of problem (\ref{policy_iteration_problem}). Moreover, if $\cP$ is finite, convergence occurs in at most $|\cP|$ iterations ($v^{|\cP|}=v^{|\cP|+1}=\cdots$).
\end{theorem}

In the context of the present problem, conditions (1) and (2) hold because the set of admissible policies $\cP$ is finite.

To verify condition (3), note that by admissibility of the control, the inventory satisfies $q\in[-\overline{Q}, \overline{Q}]$, and simultaneous buy and sell market orders are not allowed.

Consider a state $i\in\{1, \dots, M\}$ in the graph of the matrix $B(z)$ such that $d_i = 1$. Suppose $z_i=1$, which corresponds to a market buy order. Then there is an edge between nodes $i$ and $i+1$, and either $d_{i+1}=0$, in which case the required path is found, or $d_{i+1}=1$ and $z_{i+1}=1$. As we move from $i$ to $i+1$, the inventory $q$ increases by one. Repeating this argument, we eventually reach a vertex $j$ with $d_j=0$, where $q=\overline{Q}$ and a further increase in inventory is impossible.

The case $z_i=-1$, corresponding to a market sell order, is treated symmetrically.

To verify condition (4), we regroup the terms in the operator $L_{\delta}^{l^{\operatorname{a}}, l^{\operatorname{b}}}$:
\begin{gather*}
    L_{\delta}^{l^{\operatorname{a}}, l^{\operatorname{b}}} =  v(t, \alpha, q) \bigg( -\frac{k \overline{\alpha}_-}{\delta \alpha} - \frac{k \overline{\alpha}_+}{\delta \alpha} - \frac{\rho^2}{\delta \alpha^2} - \lambda^{\operatorname{a}} - \lambda^{\operatorname{b}} \bigg)
    \\
    +v(t, \alpha + \delta \alpha, q) \bigg( \frac{k \overline{\alpha}_-}{\delta \alpha} + \frac{\rho^2}{2\delta \alpha^2} \bigg) + v(t, \alpha - \delta \alpha, q) \bigg( \frac{k \overline{\alpha}_+}{\delta \alpha} + \frac{\rho^2}{2\delta \alpha^2} \bigg)
    \\
    + \lambda^{\operatorname{a}} \bigg(\1_{\{l^{\operatorname{a}}_n=0\}}\cI^+ v(t, \alpha, q) + \1_{\{l^{\operatorname{a}}_n=1\}} \cI^+ v(t, \alpha, q - 1) \bigg)
    \\
    + \lambda^{\operatorname{b}} \bigg(\1_{\{l^{\operatorname{b}}_n=0\}} \cI^- v(t, \alpha, q) + \1_{\{l^{\operatorname{b}}_n=1\}} \cI^- v(t, \alpha, q + 1) \bigg).
\end{gather*}

Hence, the diagonal entry of the matrix $I - L(w)$ has the form
\begin{equation*}
    1 + \delta t\bigg( \frac{k \overline{\alpha}_-}{\delta \alpha} + \frac{k \overline{\alpha}_+}{\delta \alpha} + \frac{\rho^2}{\delta \alpha^2} + \lambda^{\operatorname{a}} + \lambda^{\operatorname{b}} \bigg) > 0,
\end{equation*}
and thus $I - L(w)$ is indeed a strictly diagonally dominant $Z$-matrix with nonnegative diagonal elements.

The matrix $I-B(z)$ also satisfies the required properties, since its rows are either zero or contain a diagonal element equal to one and one adjacent element equal to $-1$, with all other entries equal to zero.

Therefore, the sufficient conditions of Theorem \ref{policy_iteration_convergence} are satisfied, and the policy iteration algorithm numerically yields the unique solution $v^n$ of equation (\ref{scheme}) for each $n \in \{N-1, N-2, \dots, 1\}$.

\section{Convergence of the numerical scheme}

To prove convergence of the solution of the discrete equation (\ref{scheme}) to the viscosity solution of equation (\ref{hjb}), we employ the general framework for proving convergence of finite-difference approximations to viscosity solutions of partial differential equations, developed in \cite{barlessouganidis}.

We rewrite the numerical scheme (\ref{scheme}) using the notation introduced in \cite{barlessouganidis, barlesjakobsen}:
\begin{equation}\label{scheme_BSF}
    S(\delta, x, v^\delta(x), [v^\delta]_x) = 0, \quad x \in \overline{\Omega},
\end{equation}
where $\overline{\Omega} = [0, T] \times \mathbb{R} \times ([-\overline{Q}, \overline{Q}] \cap \mathbb{Z})$,  
$S: \mathbb{R}^+ \times \overline{\Omega} \times \mathbb{R} \times C^{1,2}_b(\overline{\Omega})\rightarrow \mathbb{R}$,  
$\delta = (\delta t, \delta \alpha)$ is the grid step,  
$v^\delta: \overline{\Omega} \rightarrow \mathbb{R}$ denotes the solution of (\ref{scheme_BSF}) and an approximation of $v$,  
and $[v^\delta]_x$ coincides with $v^\delta$ at all points except $x$,
\begin{equation*}
    [v^\delta]_{x}(\overline{x}) :=
    \begin{cases}
        v^\delta(\overline{x}), & \overline{x} \neq x,\\
        0, & \overline{x} = x.
    \end{cases}
\end{equation*}

\begin{proposition}[Monotonicity]\label{monotonicity}
Let $u, w \in C^{1, 2}_b\left(\overline{\Omega}\right)$ be such that $u \geq w$. Then, for all $\delta = (\delta t, \delta \alpha) \in \mathbb{R}^+ \times \mathbb{R}^+$, $x \in \overline{\Omega}$, and $r \in \mathbb{R}$,
\begin{equation*}
    S(\delta, x, r, u) \geq S(\delta, x, r, w).
\end{equation*}
\end{proposition}

\begin{proof}
Consider a grid point $x^n_{ij}=(t_n, \alpha_i, q_j)$.  
Let $u, w\in C^{1, 2}_b(\overline{\Omega})$ be such that $u \geq w$ and $u^n_{ij}=w^n_{ij}=r$, where $u^n_{ij}=u(t_n, \alpha_i, q_j)$ and $w^n_{ij}=w(t_n, \alpha_i, q_j)$.  
Denote $[u]^n_{ij}:=[u]_{x^n_{ij}}$.  
If $t_n=T$, then
\begin{equation*}
       S(\delta, x^n_{ij}, u^n_{ij}, [u]^n_{ij}) - S(\delta, x^n_{ij}, w^n_{ij}, [w]^n_{ij}) = \Tilde{g}(T, \alpha_i, q_j) - \Tilde{g}(T, \alpha_i, q_j) = 0.
\end{equation*}
If $t_n<T$, then
\begin{multline*}
       S(\delta, x^n_{ij}, u^n_{ij}, [u]^n_{ij}) - S(\delta, x^n_{ij}, w^n_{ij}, [w]^n_{ij})
       \\
       = \max\bigg\{\sup\limits_{l_n^{\operatorname{a}}, l_n^{\operatorname{b}} \in \{0, 1\}} \bigg(\frac{u^{n+1}_{ij} - w^{n+1}_{ij}}{\delta t}+ \big(L_{\delta}^{l^{\operatorname{a}}, l^{\operatorname{b}}} (u-w)^n\big)_{ij}\bigg), \sup\limits_{z\in\{1, -1\}} \Big( \big(B_{\delta}^z (u-w)^n\big)_{ij}\Big)\bigg\} \geq 0,
\end{multline*}
because
\begin{gather*}
        \big(L_{\delta}^{l^{\operatorname{a}}, l^{\operatorname{b}}} (u-w)^n\big)_{ij} = \Big(u(t_n, \alpha_i + \delta \alpha, q_j) - w(t_n, \alpha_i + \delta \alpha, q_j)\Big)\bigg(\frac{k \overline{\alpha}_-}{\delta \alpha}+\frac{\rho^2}2\bigg)
        \\
        + \Big(u(t_n, \alpha_i - \delta \alpha, q_j) - w(t_n, \alpha_i - \delta \alpha, q_j)\Big)\bigg(\frac{k \overline{\alpha}_+}{\delta \alpha}+\frac{\rho^2}2\bigg)
        \\
        + \lambda^{\operatorname{a}} \bigg[\1_{\{l^{\operatorname{a}}_t=0\}} \cI^+ (u - w)(t_n, \alpha_i, q_j) + \1_{\{l^{\operatorname{a}}_t=1\}} \cI^+ (u - w)(t_n, \alpha_i, q_j - 1)\bigg]
        \\
        + \lambda^{\operatorname{b}} \bigg[\1_{\{l^{\operatorname{b}}_t=0\}} \cI^- (u - w)(t_n, \alpha_i, q_j) + \1_{\{l^{\operatorname{b}}_t=1\}} \cI^- (u - w)(t_n, \alpha_i, q_j + 1) \bigg] \geq 0.
\end{gather*}
\end{proof}

\begin{proposition}[Stability]\label{stability}
For any $\delta = (\delta t, \delta \alpha) \in \mathbb{R}^+ \times \mathbb{R}^+$, there exists a solution $v^\delta(t, \alpha, q)$ of equation (\ref{scheme_BSF}). Moreover, the following uniform bound holds:
\begin{equation*}
    U_1(t) \leq v^\delta(t, \alpha, q) \leq U_2(t),
\end{equation*}
where
\begin{gather*}
    U_1(t) = -\Upsilon \overline{Q} - \psi \overline{Q}^2 - (T-t)(\sigma A \overline{Q} + \phi \overline{Q}^2),
    \\
    U_2(t) = (T-t)(\Delta(\lambda^{\operatorname{a}} + \lambda^{\operatorname{b}}) + \sigma A \overline{Q}).
\end{gather*}
\end{proposition}
\begin{proof}
Existence follows from Theorem \ref{policy_iteration_convergence}.

Let $t_n < T$ and select a grid point $(\alpha_i, q_j)$ where at $t = t_n$ the value of $v^\delta$ attains its minimum:
\begin{equation*}
        v^n_{ij} = v^\delta(t_n, \alpha_i, q_j) = \min_{k, m} v^\delta(t_n, \alpha_k, q_m).
\end{equation*}
Since $\big(L_{\delta}^{l^{\operatorname{a}}, l^{\operatorname{b}}} v^n\big)_{ij} \geq 0$ for all $l^{\operatorname{a}}_n, l^{\operatorname{b}}_n \in \{0, 1\}$,
\begin{gather*}
        0 = \max\Bigg\{\frac{v_{ij}^{n+1} - v_{ij}^n}{\delta t} + 
        \sup\limits_{l_n^{\operatorname{a}}, l_n^{\operatorname{b}} \in \{0, 1\}} \Big(\big(L_{\delta}^{l^{\operatorname{a}}, l^{\operatorname{b}}} v^n\big)_{ij} + \Tilde{f}_{ij}^{l^{\operatorname{a}}, l^{\operatorname{b}}, n} \Big),
        \sup\limits_{z\in\{1, -1\}} \Big( \big(B_{\delta}^z v^n\big)_{ij} - \Upsilon \Big)
        \Bigg\} \geq
        \\
        \geq \frac{v_{ij}^{n+1} - v_{ij}^n}{\delta t} - \sigma A \overline{Q} - \phi \overline{Q}^2.
\end{gather*}
Inductively in $n$ we obtain
\begin{gather*}
        v_{ij}^n \geq v_{ij}^{n+1} - \sigma A \overline{Q} - \phi \overline{Q}^2 \geq \min_{k, m} v^\delta(t_{n+1}, \alpha_k, q_m) - \sigma A \overline{Q} - \phi \overline{Q}^2.
\end{gather*}
Since $v^\delta(T, \alpha, q) = \Tilde{g}(T, \alpha, q) \geq -\Upsilon \overline{Q} - \psi \overline{Q}^2$, we get the lower bound
\begin{equation*}
        v^\delta(t, \alpha, q) \geq -\Upsilon \overline{Q} - \psi \overline{Q}^2 - (T-t)(\sigma A \overline{Q} + \phi \overline{Q}^2) = U_1(t).
\end{equation*}
The lower bound corresponds to the case where the market maker holds the extreme position $\pm \overline{Q}$ during the entire period $[t, T]$, while the asset price follows a strong adverse trend, producing the largest running and terminal penalties.

Similarly, for the upper bound let $t_n < T$ and take $(\alpha_i, q_j)$ where $v^\delta(t_n, \alpha_i, q_j)$ attains its maximum:
\begin{equation*}
        v^n_{ij} = v^\delta(t_n, \alpha_i, q_j) = \max_{k, m} v^\delta(t_n, \alpha_k, q_m).
\end{equation*}
Then $\big(L_{\delta}^{l^{\operatorname{a}}, l^{\operatorname{b}}} v^n\big)_{ij} \leq 0$ for $l_n^{\operatorname{a}}, l_n^{\operatorname{b}} \in \{0, 1\}$, and
\begin{equation*}
        \big(B_{\delta}^z v^n\big)_{ij} - \Upsilon \leq -\Upsilon < 0 \quad \text{for all } z\in \{1, -1\}.
\end{equation*}
Hence,
\begin{gather*}
        0 = \max\Bigg\{\frac{v_{ij}^{n+1} - v_{ij}^n}{\delta t} + 
        \sup\limits_{l_n^{\operatorname{a}}, l_n^{\operatorname{b}} \in \{0, 1\}} \Big(\big(L_{\delta}^{l^{\operatorname{a}}, l^{\operatorname{b}}} v^n\big)_{ij} + \Tilde{f}_{ij}^{l^{\operatorname{a}}, l^{\operatorname{b}}, n} \Big),
        \sup\limits_{z\in\{1, -1\}} \Big( \big(B_{\delta}^z v^n\big)_{ij} - \Upsilon \Big)
        \Bigg\} \\
        \leq \frac{v_{ij}^{n+1} - v_{ij}^n}{\delta t} + \Delta(\lambda^{\operatorname{a}} + \lambda^{\operatorname{b}}) + \sigma A \overline{Q}.
\end{gather*}
This gives the induction step
\begin{equation*}
        v_{ij}^n \leq v_{ij}^{n+1} + \Delta(\lambda^{\operatorname{a}} + \lambda^{\operatorname{b}}) + \sigma A \overline{Q} \leq \max_{k, m} v^\delta(t_{n+1}, \alpha_k, q_m) + \Delta(\lambda^{\operatorname{a}} + \lambda^{\operatorname{b}}) + \sigma A \overline{Q}.
\end{equation*}
Since $v^\delta(T, \alpha, q) = \Tilde{g}(T, \alpha, q) \leq 0$, we obtain the upper bound
\begin{equation*}
        v^\delta(t, \alpha, q) \leq (T-t)(\Delta(\lambda^{\operatorname{a}} + \lambda^{\operatorname{b}}) + \sigma A \overline{Q}) = U_2(t).
\end{equation*}
The upper bound $U_2(t)$ can be interpreted as the accumulation of two maximum gains over $[t, T]$: from the spread and from exploiting predictable price movements.
\end{proof}

\begin{proposition}[Consistency]\label{consistency}
For all $(t, \alpha, q) \in \overline{\Omega}$ and $\varphi \in C^{1, 2}_b (\overline{\Omega})$,
\begin{multline*}
    \lim_{\substack{(\delta t, \delta \alpha) \to (0, 0)\\ (t', \alpha') \to (t, \alpha) \\ \xi \to 0}} S\Big((\delta t, \delta \alpha), (t', \alpha', q), \varphi^\delta(t', \alpha', q)+\xi, [\varphi^\delta+\xi]_{(t', \alpha', q)}\Big) 
    \\
    =\max\Bigg(
    \sup\limits_{l_n^{\operatorname{a}}, l_n^{\operatorname{b}} \in \{0, 1\}} \bigg(\frac{\partial \varphi}{\partial t} + L^{l^{\operatorname{a}}, l^{\operatorname{b}}} \varphi + \Tilde{f}^{l^{\operatorname{a}}, l^{\operatorname{b}}} \bigg),
    \sup\limits_{z\in\{1, -1\}} \Big( \cM^z \varphi - \varphi \Big)
    \Bigg).
\end{multline*}
\end{proposition}
\begin{proof}
From the definition of the numerical scheme and the continuity of $\varphi$, we obtain
\begin{gather*}
\lim_{\substack{(\delta t, \delta \alpha) \to (0, 0)\\ (t', \alpha') \to (t, \alpha)}}\sup\limits_{l_n^{\operatorname{a}}, l_n^{\operatorname{b}} \in \{0, 1\}} \bigg(\frac{\varphi(t'+\delta t, \alpha', q) - \varphi(t', \alpha', q)}{\delta t} + L_{\delta}^{l^{\operatorname{a}}, l^{\operatorname{b}}} \varphi(t', \alpha', q) + \Tilde{f}^{l^{\operatorname{a}}, l^{\operatorname{b}}}(t', \alpha', q) \bigg)
\\
=\sup\limits_{l_n^{\operatorname{a}}, l_n^{\operatorname{b}} \in \{0, 1\}} \bigg(\frac{\partial \varphi}{\partial t}(t, \alpha, q) + L^{l^{\operatorname{a}}, l^{\operatorname{b}}} \varphi(t, \alpha, q) + \Tilde{f}^{l^{\operatorname{a}}, l^{\operatorname{b}}}(t, \alpha, q) \bigg),
\end{gather*}
and
\begin{equation*}
\lim_{\substack{(\delta t, \delta \alpha) \to (0, 0)\\ (t', \alpha') \to (t, \alpha)}}\sup\limits_{z\in\{1, -1\}} \Big( B_{\delta}^z \varphi(t', \alpha', q) - \Upsilon \Big)
=\sup\limits_{z\in\{1, -1\}} \Big( \cM^z \varphi(t, \alpha, q) - \varphi(t, \alpha, q) \Big).
\end{equation*}
\end{proof}

We rely on the comparison principle from \cite[Theorem~2]{carteawang}.
\begin{theorem}[Comparison principle]
Let $v_1$ and $v_2$ be bounded subsolution and supersolution of equation (\ref{hjb}), respectively, and assume that $v_1(T) \leq v_2(T)$. Then $v_1 \leq v_2$.
\end{theorem}

\begin{theorem}[Convergence]
As $(\delta t, \delta \alpha) \to (0, 0)$, the solution $v^{\delta}(t, \alpha, q)$ of equation (\ref{scheme_BSF}) converges locally uniformly to the unique viscosity solution $v(t, \alpha, q)$ of equation (\ref{hjb}).
\end{theorem}
\begin{proof}
For $(t, \alpha, q) \in \overline{\Omega}$, define
\begin{equation*}
    \underline{v}(t, \alpha, q) = \liminf_{\substack{(\delta t, \delta \alpha) \to (0, 0)\\ (t', \alpha') \to (t, \alpha)}} v^{\delta t, \delta \alpha}(t', \alpha', q),
    \quad
    \overline{v}(t, \alpha, q) = \limsup_{\substack{(\delta t, \delta \alpha) \to (0, 0)\\ (t', \alpha') \to (t, \alpha)}} v^{\delta t, \delta \alpha}(t', \alpha', q).
\end{equation*}
By definition, $\underline{v} \leq \overline{v}$. Moreover, $\underline{v}(T, \alpha, q) = \overline{v}(T, \alpha, q)$ for all $\alpha, q \in \mathbb{R} \times ([-\overline{Q}, \overline{Q}] \cap \mathbb{Z})$.  
The boundedness of these functions follows from the stability result (Proposition \ref{stability}).  
Hence, to establish the reverse inequality $\underline{v} \geq \overline{v}$, it suffices—by applying the comparison principle—to verify that $\underline{v}$ is a viscosity supersolution and that $\overline{v}$ is a viscosity subsolution.  
We prove the first statement; the second follows by symmetry.

Let $\varphi \in C_b^{1, 2}\big(\overline{\Omega}\big)$ and suppose that $(\tilde{t}, \tilde{\alpha}, \tilde{q})$ is a global minimum point of $\underline{v} - \varphi$.  
Without loss of generality, assume that the minimum is strict and that $\underline{v}(\tilde{t}, \tilde{\alpha}, \tilde{q}) = \varphi(\tilde{t}, \tilde{\alpha}, \tilde{q})$.

Then there exist sequences $\delta_k = (\delta t_k, \delta \alpha_k) \in \mathbb{R}^+ \times \mathbb{R}^+$ and $(t_k, \alpha_k, q_k) \in \overline{\Omega}$ such that, as $k \to \infty$,
\begin{equation*}
    (\delta t_k, \delta \alpha_k) \to (0, 0), \quad (t_k, \alpha_k, q_k) \to (\tilde{t}, \tilde{\alpha}, \tilde{q}), \quad v^{\delta_k}(t_k, \alpha_k, q_k) \to \underline{v}(\tilde{t}, \tilde{\alpha}, \tilde{q}),
\end{equation*}
and $v^{\delta_k} - \varphi$ attains a global minimum at $(t_k, \alpha_k, q_k)$.

Let $\xi_k = v^{\delta_k}(t_k, \alpha_k, q_k) - \varphi(t_k, \alpha_k, q_k)$.  
Then $v^{\delta_k}(t, \alpha, q) \geq \varphi(t, \alpha, q) + \xi_k$ for all points $(t, \alpha, q) \in \overline{\Omega}$, with $\xi_k \to 0$ as $k \to \infty$.

By monotonicity (Proposition \ref{monotonicity}) of the scheme $S$, we have
\begin{eqnarray*}
    0 &=& S\Big(\delta_k, (t_k, \alpha_k, q_k), v^{\delta_k}(t_k, \alpha_k, q_k), [v^{\delta_k}]_{(t_k, \alpha_k, q_k)}\Big) 
    \\
   & \geq& S\Big(\delta_k, (t_k, \alpha_k, q_k), v^{\delta_k}(t_k, \alpha_k, q_k), [\varphi + \xi_k]_{(t_k, \alpha_k, q_k)}\Big) 
    \\
   & = &S\Big(\delta_k, (t_k, \alpha_k, q_k), \varphi(t_k, \alpha_k, q_k) + \xi_k, [\varphi + \xi_k]_{(t_k, \alpha_k, q_k)}\Big).
\end{eqnarray*}
Thus,
\begin{equation*}
    S\Big(\delta_k, (t_k, \alpha_k, q_k), \varphi(t_k, \alpha_k, q_k) + \xi_k, [\varphi + \xi_k]_{(t_k, \alpha_k, q_k)}\Big) \leq 0.
\end{equation*}
Applying the consistency property (Proposition \ref{consistency}) and taking the limit as $k \to \infty$, we obtain
\begin{equation*}
    \max\Bigg(
    \sup\limits_{l_n^{\operatorname{a}}, l_n^{\operatorname{b}} \in \{0, 1\}} \bigg(\frac{\partial \varphi}{\partial t} + L^{l^{\operatorname{a}}, l^{\operatorname{b}}} \varphi + \Tilde{f}^{l^{\operatorname{a}}, l^{\operatorname{b}}} \bigg),
    \sup\limits_{z\in\{1, -1\}} \Big( \cM^z \varphi - \varphi \Big)
    \Bigg) \leq 0.
\end{equation*}
Therefore, $\underline{v}$ is a viscosity supersolution and $\overline{v}$ is a viscosity subsolution.  
By the comparison principle and the definitions of $\underline{v}$ and $\overline{v}$, it follows that $\underline{v} \equiv \overline{v}$, and this common function is the unique viscosity solution of equation (\ref{hjb}).
\end{proof}

\section{Numerical experiment}

We solve the Hamilton--Jacobi--Bellman Quasi-Variational Inequality numerically with the following parameter values:
\begin{gather*}
    T=10, \ A=300, \ \overline{Q} = 4, \ \sigma=0.01, \ \theta=0.1,
    \\
    \Delta = 0.005, \ \e = 0.005, \ \lambda^{\operatorname{a}}=\lambda^{\operatorname{b}}=1, \ k=200,
    \\
    \rho=1, \ \gamma^{\operatorname{a}}=\gamma^{\operatorname{b}}=60, \ \phi=10^{-6}, \ \psi =0.
\end{gather*}
Let the uniform grid in $\alpha$ consist of $N_\alpha=101$ points, and let the uniform time grid contain $N=200$ points.

\begin{figure}[h]
    \begin{center}
        \begin{minipage}[h]{0.45\linewidth}
            \includegraphics[width=1\linewidth]{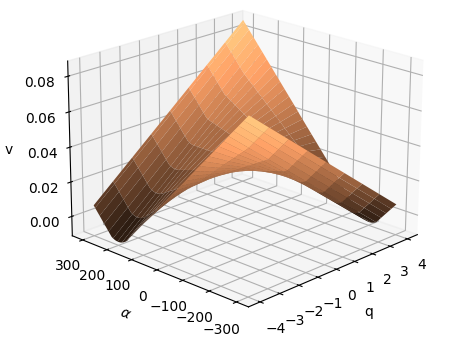}
            \caption{Surface of the value function $v$ at time $t=0$.}
            \label{pic:value_function}
        \end{minipage}
        \hfill
        \begin{minipage}[h]{0.5\linewidth}
            \includegraphics[width=1\linewidth]{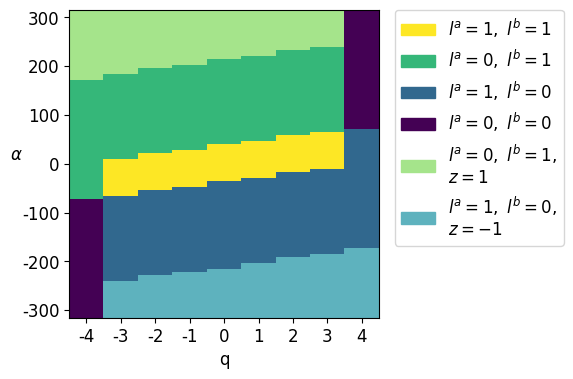}
            \caption{Optimal control at time $t=0$.}
            \label{pic:optimal_control}
        \end{minipage}
    \end{center}
\end{figure}

Figure~\ref{pic:value_function} shows the computed surface of the value function $v$ at the initial time.  
The corresponding optimal control is displayed in Figure~\ref{pic:optimal_control}.  
From these results it can be observed that when the alpha signal is close to zero, the market maker quotes both bid and ask limit orders.  
As the alpha signal increases (or decreases), the market maker begins to trade only on the buy (or sell) side.  
For large values of the alpha signal, the market maker uses both limit and market orders to exploit predictable price movements efficiently.

\section{Conclusion}

This work investigates an implicit numerical method for solving the Hamilton--Jacobi--Bellman Quasi-Variational Inequality arising in combined stochastic and impulse control problem for a market maker.  
The proposed approach avoids time-step restrictions due to its unconditional stability.

We established the convergence of the policy iteration algorithm to the solution of the discrete scheme at each time step.  
Furthermore, the convergence of the implicit numerical solution to the unique viscosity solution of the Hamilton--Jacobi--Bellman Quasi-Variational Inequality was proved.  
To this end, we verified the monotonicity, stability, and consistency properties of the numerical scheme.

A numerical experiment was conducted, and the results illustrate the shape of the value function and the structure of the optimal control at the initial time.

\section*{Acknowledgements}
The author is thankful to Yuri Kabanov for the attention to his work.


\renewcommand{\refname}{References}

\end{document}